\documentclass{llncs}
\usepackage{floatpag}
\floatpagestyle{empty} 

\usepackage{graphicx} % required for inserting images
\usepackage[ruled,vlined,linesnumbered,noresetcount]{algorithm2e}

\usepackage{comment}  % XXX: remove if we no longer need it

\usepackage[hidelinks]{hyperref}

\newtheorem{assumption}{Assumption}

\let\llncssubparagraph\subparagraph
\let\subparagraph\paragraph
\usepackage{titlesec}
\let\subparagraph\llncssubparagraph
\titleformat{\subsection}[runin]
    {\normalfont\bfseries}{}{0em}{}[.]

\usepackage{xcolor}
\definecolor{lightred}{RGB}{255,130,130}

\def\Byz{\ensuremath{f}}   % the Byzantine node
\def\ByzN{\ensuremath{x}}  % number of Byzantine nodes
\def\R{\ensuremath{\mathsf{R}}}  % right-hand-rule
\def\L{\ensuremath{\mathsf{L}}}  % left-hand-rule
\usepackage{enumitem}
\setlist[description]{font=\normalfont\itshape\/, labelindent=0.5cm}

\title{BeRGeR: Byzantine-Robust Geometric Routing}
\author{Brown Zaz \and Mikhail Nesterenko \and Gokarna Sharma}
\institute{Department of Computer Science, Kent State University, Kent, OH 44242, USA\\
\email{\{zbrown, mikhail, sharma\}@cs.kent.edu}}
\date{}

\begin{document}
\sloppy
\maketitle
\thispagestyle{plain}
\pagestyle{plain}

\begin{abstract}
  We present BeRGeR: the first asynchronous geometric routing algorithm that guarantees delivery of a message despite a Byzantine fault without relying on cryptographic primitives or randomization. The communication graph is a planar embedding that remains three-connected if all edges intersecting the source-target line segment are removed.  We prove the algorithm correct and estimate its message complexity.
\end{abstract}

\section{Introduction}
Geometric, also called geographic, routing uses node locations to transmit a message from the source to the target node. Nodes may either obtain their coordinates from GPS or compute virtual coordinates~\cite{IN99,KV99,KV00,KSU99}. Geometric routing has several attractive features. Nodes do not need to store network topology information beyond their immediate neighbors. Moreover, such routing may be stateless as nodes do not have to retain information after forwarding a packet. Compared to flooding-based ad hoc routing algorithms~\cite{dsr,aodv}, geometric routing is resource efficient.

% may be used for bootstrapping more sophisticated routing schemes, MN
Geometric routing may therefore be used in cases where maintaining more extensive routing information is not practicable.
It may operate in environments with high topological volatility such as vehicular networks~\cite{vehicularNetworking} or large collections of resource-poor devices such as wireless sensor networks~\cite{lites}.

Despite the claim of hostile environment applicability, little research has been done on fortifying geometric routing against faults using geometric routing techniques themselves. In this paper, we address this issue.

\subsection{Geometric routing}
The simplest form of geometric routing is greedy. In greedy routing, the packet is forwarded to the neighbor with the closest Euclidean distance to the target. However, greedy routing fails if it reaches a local minimum. A local minimum is a node that does not have an immediate neighbor closer to the target. To recover from a local minimum and guarantee message delivery, packets are sent to traverse faces of a planarized subgraph~\cite{BMSU01}. Finding a maximum planar subgraph of a general graph is NP-hard~\cite{NPHARD}. However, for certain graphs, the task may be solved efficiently. A graph is unit-disk if a pair of its nodes \(u\) and \(v\) are neighbors if and only if the Euclidean distance between them is no more than \(1\). Such a graph approximates a wireless network. In a unit-disk graph, a connected planar subgraph may be found by local computation at every node using Relative Neighborhood or Gabriel Graph~\cite{GFG,GS69,GPSR,T80}.

A sequential geometric routing algorithm, such as the classic GFG/GPSR~\cite{GFG,GPSR}, routes a single packet in the greedy mode until a local minimum is encountered. The algorithm then switches to recovery mode, which involves traversing the faces of a planar subgraph of the original communication graph. Specifically, the algorithm traverses the faces that intersect the line segment that connects the source and the target.

Algorithms such as GFG/GPSR have a problem of face traversal direction choice. A face may potentially be traversed in two directions: clockwise and counter-clockwise. Face traversal may be inefficient if its traversal direction is selected inappropriately: the traversal distance may be long in one direction and short in the other. GOAFR+ \cite{GOAFR+} finds the shorter traversal direction by reversing it once the packet reaches a pre-determined ellipse containing source and target nodes.

A concurrent geometric routing algorithm CFR~\cite{CFR} sends multiple packets to traverse the faces intersecting the source-target line in both directions in parallel. This naturally selects the shortest face traversal direction. 

\subsection{Byzantine fault tolerance}
A Byzantine node~\cite{lamport1982byzantine,pease1980reaching} may behave arbitrarily. This is the strongest fault that can affect a node in a distributed system. A reputation-based approach~\cite{sanchez2016privhab+,pathak2008securing} is considered to deal with Byzantine faults. However, a faulty node may actively resist reputation compromise. Hence, in general, such approaches may only alleviate rather than eliminate the problem.
The power of the faults may also be mitigated with cryptography~\cite{lamport1982byzantine,dolev1983authenticated,katz2009expected} or randomization~\cite{ben1983another,feldman1997optimal}. 
Synchrony assumptions may help with fault handling as well~\cite{PBFT}: if packet transmission may be delayed only for a finite amount of time, then fault information may be obtained from lack of packet receipt. In a completely asynchronous system, such information is not available.

Cryptography may be too expensive for resource poor nodes, the source of true randomness may not be achievable and synchrony may be impossible to guarantee. If none of these primitives are available, the solution requires that the number of correct processes be large enough to overwhelm the faulty ones.

The complexity of Byzantine fault handling increases if the network is not completely connected. In this case, nodes may not communicate directly; they have to rely on intermediate nodes to forward the packets. Faulty nodes may tamper with such forwarding. To counter such faulty behavior, packets are sent along alternative routes. To enable this, the network should be sufficiently connected. In general topology, message transmission is possible only if the network is \(2\ByzN+1\)-connected~\cite{lamport1982byzantine,dolev1982byzantine,fischer1986easy}, where \(\ByzN\) is the maximum number of Byzantine faults.

\ \\
In this paper, we study Byzantine-robust geometric routing in asynchronous networks. We do not use cryptography, randomization or reputation. Instead, we use distributed geometric routing to bypass the faults.

\subsection{Related work}

Sanchez et al~\cite{sanchez2016privhab+} and Pathak et al~\cite{pathak2008securing} use both cryptographic and reputation-based approaches to secure geometric routing. Adnan et al~\cite{adnan2017secure} propose to secure geometric routing through pairwise key distribution. 
Boulaiche and Bouallouche~\cite{boulaiche2017hsecgr} use message authentication codes to prevent message tampering. Maurer and Tixeuil~\cite{maurer2014containing} consider containing faulty Byzantine nodes in control zones such that messages will be sent there only with the authentication form border nodes. Zahariadis et al~\cite{zahariadis2013novel} propose a reputation-based approach to geometric routing security.
Several papers discuss counteracting spurious locations reported by faulty nodes to secure geometric routing~\cite{garcia2010secure,leinmuller2006improved,vora2006secure}.
Recently, the problem of consensus has been explored in the geometric setting~\cite{oglio2021byzantine}. 
Zaz and Nesterenko~\cite{BeRGeRJMM} consider Byzantine-robust geometric routing but offer no solution to the problem. 
To the best of our knowledge, no Byzantine-robust asynchronous geometric routing algorithm without cryptography, randomization or reputation has been proposed.

Several related problems have been addressed with concurrent geometric routing. We consider unicasting: sending a message from a single source to a single target.
Alternatively, in multicasting, the same message is delivered to a set of nodes in the network~\cite{SRS06}. Sequential geometric multicasting algorithms~\cite{PBM,GMP,LGS} optimize message transmission routes by forwarding the same packet to multiple targets for a part of the route.
MCFR~\cite{adamek2018concurrent} concurrently sends packets along all the appropriate faces achieving faster delivery at the expense of a greater number of transmitted packets.
Another related problem is geocasting; in this problem, the source needs to deliver messages to every node in a particular target area. There are several sequential geocasting algorithms~\cite{GFG,SRS06,LNLC06}. Adamek et al~\cite{adamek2017stateless} present a concurrent geocasting solution. None of these algorithms consider Byzantine tolerance.

\subsection{Our contribution}
We present BeRGeR, an asynchronous unicast concurrent geometric routing algorithm that handles a single Byzantine fault. We assume 
the source and target nodes are connected by three internally node-disjoint paths that do not intersect the source-target line, formally prove BeRGeR correct under this assumption, and analyze its message complexity. 

\section{Preliminaries}

\subsection{Communication model}
A finite connected graph \(G\) is embedded in a geometric plane. Two nodes may not share the same coordinates. We, therefore, use node coordinates for both navigation and node identification. Two nodes adjacent to the same edge are \emph{neighbors}
% is that prof. Sharma? no multigraphs! MN
%{\color{green}and share only a single edge}.
Each node knows its own coordinates and the
coordinates of its neighbors.
Neighbors communicate by passing packets. The communication is bi-directional, so the graph, \(G\), is undirected. The packet transmission is FIFO, reliable and asynchronous.

Nodes are either correct or faulty. A \emph{correct} node operates according to the specified algorithm. A \emph{faulty} node is Byzantine: it may behave arbitrarily including sending packets, dropping received packets or not communicating at all. We use the standard assumption that the Byzantine node may not forge its own identifier, i.e. coordinates.
% this is a pretty standard assumption for Byzantine tolerance without certificates: otherwise the problem is intractable, MN
% forging coordinates in radio networks is a separate research field. We cite a few papers on that in lit review, MN
%{\color{blue} GS: does it mean the adversary is ``weak'' (there are separations on weak and strong adversaries)?}.
%
That is, the packet recipient always correctly identifies the sending neighbor. 

If a node sends a packet without receiving it first, the node \emph{originates} the packet.  A node \emph{forwards} a packet if it receives it and then sends it to another node. If a node forwards a packet to more than one neighbor, the node \emph{splits} the packet.
If a node receives but does not forward the packet, it \emph{drops} the packet. We assume that a faulty node drops all packets that it receives and originates all packets that it sends. The case of a faulty node forwarding a packet correctly is equivalent to dropping the packet and originating an identical packet. 

A \emph{message} is the gainful content to be transmitted by a sequence of forwarded packets. An arbitrary \emph{source} node \(s \in G\) is to transmit a message \(m_s\) to another arbitrary \emph{target} node \(t \in G\). To simplify the presentation, we assume that $s$ and $t$ are not neighbors. 
For our purposes, the message content is immaterial provided that two messages can be compared for equality. 
% as per Gokarna, MN
The target may receive the message from multiple neighbors, the target \emph{delivers} the message once it passes correctness checks. 
Forwarding a packet with the same message creates a message path. Consider two message paths from node \(u\) to node \(v\). These paths are \emph{internally node-disjoint} if the only nodes the paths share are \(u\) and \(v\).

\subsection{Planarity, faces, traversal}
An embedding of a graph is \emph{planar} if its edges intersect only at nodes. For short, we call such
a planar embedding a \emph{planar graph}. In a planar graph \(G\), a \emph{face} is a region on the plane such that any pair of its points can be connected by a continuous curve inside the face. If the graph is finite, then the area of all but one face is finite. The finite area faces are \emph{internal} faces and the infinite area face is the \emph{external} face.

To traverse a face of a planar graph, the packets are routed using right- or left-hand-rule. Consider a node \(v \in G\) that receives a packet from its neighbor \(u\).
In the \emph{right-hand-rule}, \(v\) forwards the packet to the next clockwise neighbor after \(u\).
In the \emph{left-hand-rule}, \(v\) forwards the packet to the next counter-clockwise neighbor after \(u\).
We call the obtained traversal paths respectively \emph{right} and \emph{left} and denote them \R\ and \L.
The right path traverses an internal face counter-clockwise and the external face clockwise. The left path traverses the internal and external faces in the opposite direction to the right path. 

Let $G - \overline{st}$ be the graph, $G$, without the edges intersecting the source-target line segment $\overline{st}$. 
The \emph{green face} is the union of faces that intersect \(\overline{st}\). In other words, the green face is the unique face that contains segment \(\overline{st}\) in \( G - \overline{st} \). This green face in \( G - \overline{st} \) may be either internal or external. A \emph{green node} is a node adjacent to the green face. The source and target are thus green nodes. Given a green node \(k\), the \(k\)-\emph{blue face} is a union of non-green faces adjacent to \(k\). A $k$-\emph{blue} node is a non-green node adjacent to the $k$-blue face.

\begin{figure}[htb]
\centering
\includegraphics[angle=-90,width=.50\linewidth]{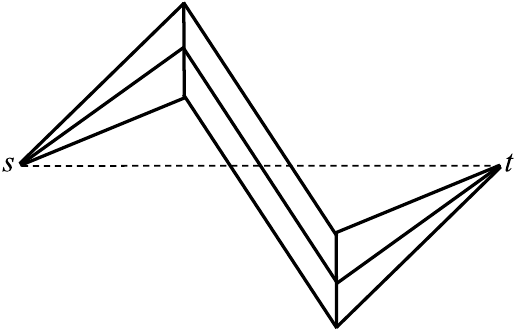}
\caption{A graph, \(G\), violating \nameref{G-st_3connected}. Even though the graph itself is three-connected, \(G - \overline{st}\) is disconnected.}\label{figWorm}
\end{figure}

\subsection{The Byzantine-Robust Geometric Routing Problem}
 An algorithm that solves the Byzantine-Robust Geometric Routing Problem delivers the message \(m_s\) sent by the source to the target subject to the following three properties:
\begin{description}
    \item[Validity: \label{validityDescription}] if the target delivers
    % I addressed this above in the defintion of receipt and delivery, MN
    % {\color{blue} GS: what do we mean by ``target delivers (this and next bullet)''?} 
    message \(m_t\), then \(m_t = m_s\);
    
    \item[Liveness: \label{livenessDescription}] the target eventually delivers \(m_t\);
    \item[Termination: \label{terminationDescription}] every packet is forwarded a finite number of times.
\end{description}

\subsection{Fault and connectivity assumptions}
We consider a solution to the problem with at most one faulty node \(\Byz \in G\). We assume that $s$ and $t$ are correct. 
%{\color{blue} GS: seems rearrangement is needed here}
% done, MN

% TODO: "the packet nedes to travel by at lesat 3 disjoint paths". This could use some clarification. 3 disjoint paths need to exist, but the packet does not need to travel along them. In our algorithm, the paths our packets takes are not disjoint. -Zaz
% fixed, MN
The Byzantine node may originate a spurious message or stop the correct message from propagating. Thus, there need to be at least 3 node-disjoint paths to bypass the faulty node~\cite{dolev1982byzantine}. Hence, we consider 3-connected graphs. 

Finding internally node-disjoint paths is difficult even if the graph is 3-connected. Our idea is to spatially separate them: one path uses the left-hand-rule traversal of the green face \(G - \overline{st}\), the other --- right-hand-rule. To achieve this, when forwarding the message, each green node ignores the edges that intersect \(\overline{st}\). However, in general, this may eliminate potential paths to the target or leave the graph entirely disconnected. See \autoref{figWorm} for an example. To prevent such disconnect, we posit the following graph connectivity assumption:

\begin{assumption}[the Triconnectivity Assumption]\label{G-st_3connected}
  \( G - \overline{st} \) is 3-connected.
\end{assumption}

\section{BeRGeR Description}
% TODO: MN, you said you wanted to switch "Byzantine" for "faulty". Shall we start doing that here?
% this is a good place to start, MN
In this section we present BeRGeR:   
% {\color{gray}a} {\bf GS: the first?}
% eek, we say the algorithm is the first in the abstract. I prefer not to push it this hard mid-paper, MN
a Byzantine-Robust Geometric Routing algorithm that solves the Reliable Message Delivery Problem with a single faulty node in a connected planar graph subject to \nameref{G-st_3connected}.

\subsection{Algorithm outline}
The algorithm operates as follows. The source concurrently sends two packets to traverse the green face in the opposite traversal directions. These packets are \emph{cores}. That is, the source sends a left core and a right core.

A green node, i.e. a node adjacent to the green face, may be faulty. The faulty node may send a core with a spurious message. To prevent this, the target waits to receive both cores with matching messages before delivering their message. In this case, however, a faulty green node may drop the packet altogether and prevent message delivery at the target. To counteract this, BeRGeR has a mechanism of bypassing, or skipping, every green node. As the green nodes forward the core along the border of the green face, they add the nodes that the core visited to the packet. In addition, each green node splits the core by sending a \emph{thread} packet that skips the next green node $k$. This thread packet traverses the union of \(k\)-blue and green faces.
% I have seen papers with coled pictures,
% but papers with colored text is the next step in colorful delivery, MN
%{\bf GS: confused here, what is k-blue? i think this k-blue needs a definition? UPDATE: i found definition now, may be if we use green and blue colors in paper, would provide more clarity}

A \emph{braid} is a set of threads that carry matching messages 
% added this, MN
in the same direction (\L\ or \R). 
% TODO: added (\L\ or \R) to remind the reader what direction means -Zaz
A braid \emph{matches} a core if it contains threads that skip each node that the core has visited and every such thread carries the same message as the core.
This way, if there is a faulty green node, at least one thread skips it. Therefore, if a faulty node attempts to drop a packet or originate a forged packet, the target does not collect a matching braid and core. Since the faulty node may be only on the left or right side of the green face, a matching braid and core arrive on the opposite side regardless of the faulty node's actions.

\begin{figure}[ht]
\centering
\includegraphics[angle=-90,width=.77\linewidth]{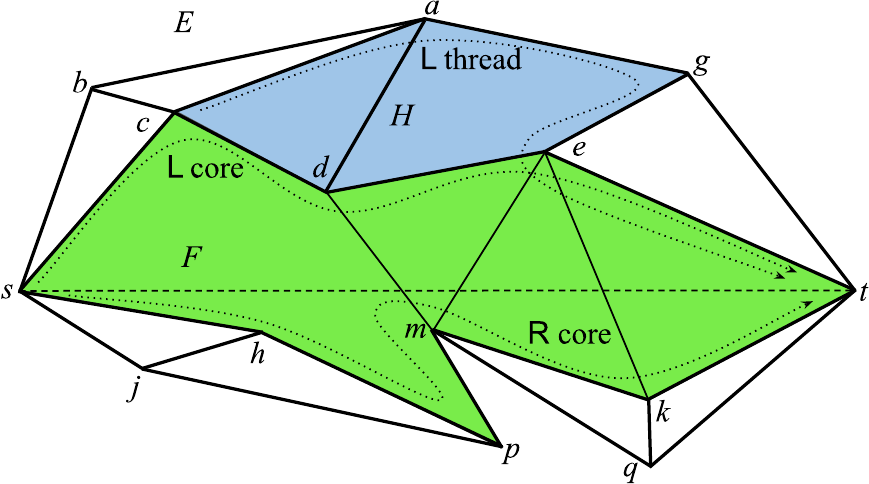}
\caption{BeRGeR example operation. \(\L\) core and \(\R\) core traverse the green face, \(F\). The second \(\mathsf{L}\) thread, the \(d\)-thread, skips node \(d\) and traverses the union of the \(d\)-blue face, \(H\), and the green face \(F\).
%{\color{blue} GS: can we in fact use green and blue colors in the figure?}
% done, MN
}
\label{figBerger}
\end{figure}

Refer to \autoref{figBerger} for an illustration of the algorithm operation. The planar graph shown in the picture complies with \nameref{G-st_3connected}. Note that greedy routing fails on this graph if the packet reaches node \(p\) since \(p\) is a local minimum. Face $F$ is green. Nodes \(s,c,d,e,t,k,m,p\) and \(h\) are adjacent to $F$. These nodes are, therefore, green.
Face \(H\) is the \(d\)-blue face adjacent to \(F\). Nodes \(c, a, g, e\)  are $d$-blue since they are not green but adjacent to the $d$-blue face, $H$.  Face \(E\) is the external face of the graph. The source sends two core packets that traverse \(F\): left core \(\L\) and right core \(\R\). The source also sends a left and a right thread that are not shown.
The figure illustrates a single $\L$ thread that is generated by $c$ when $c$ receives a core. That is, $c$ splits this core and sends thread that skips the next green node $d$. This thread traverses part of \(H\) and then \(F\) to reach $t$.

\SetKw{and}{\ and\ }
\SetKw{or}{\ or\ }

\SetKwFunction{onThisSide}{onThisSide}

\SetKwData{R}{R}
\SetKwData{L}{L}
\SetKwData{true}{true}
\SetKwData{false}{false}

\begin{algorithm}[htb]

\footnotesize
\caption{BeRGeR variables and  functions}\label{algBERGERfunctions}

% \texttt to make the font match algorithm2e's \SetKwFunction
\SetKwProg{nextNode}{\(\texttt{nextNode}(p, s, t, c, k)\)}{}{}
\SetKwProg{onThisSideOf}{\(\texttt{onThisSide}(s, t, n, i)\)}{}{}
\SetKwProg{invalid}{\(\texttt{invalid}(p, m, s, t, k, \ell)\)}{}{}

\textbf{constants} \\
\ \ \ \ \(n\) \ \ \tcp*[f]{node coordinates} \\
\ \ \ \ \(N\) \ \ \tcp*[f]{set of neighbor coordinates} \\
\ \ \ \ \(\{ \L,\R \}\) \ \ \tcp*[f]{left and right packet direction} \\
\vspace{2mm}

\textbf{variables} \\
\ \ \ \ \(T\) \tcp*[f]{set of packets received by target} \\
\vspace{2mm}

\textbf{functions}\\
\vspace{1mm}
\nextNode(){}{
     \(M := \{i \in N : i\neq k \and \onThisSide(s,t,n,i)\}\) \\
     \Return \(i \in M : i\) is nearest to \(p\) in direction  \(c \in \{\L, \R\} \) \\
      \tcp{node selection direction \L: counter-clockwise, \R: clockwise}
      % TODO: add to description: \tcp{guaranteed to return due to \nameref{G-st_3connected}}
}

\vspace{2mm}

\onThisSideOf(){}{
    \tcp{return \true if the \(n\)-\(i\) and \(s\)-\(t\) line segments do not intersect}
    \Return \( \overline{ni} \cap \overline{st} = \emptyset \)
}

\vspace{2mm}

\invalid(){}{ \label{lineInvalid}
  \If(\tcp*[f]{traveled in a cycle}){\( k = p \)\tcp*[f]{sender is skipping itself}\label{lineSkipsItself} \\
  % the below indicates the Byzantine node is on the green face OR the connectivity assumptions are violated
  \or \( \left( n = s \and k = \bot \right) \) \tcp*[f]{a core arrives at \(s\)}\label{lineCoreDroppedByS} \\
  \or \( n \in \ell \) }
  { \Return \true }
  \Else{\Return \false}
}\label{endOfAlgOne}
\end{algorithm}

\begin{algorithm}[htbp]
\footnotesize

\SetKw{from}{from}
\SetKw{deliver}{deliver }
\SetKw{receive}{receive}

\SetKwFor{Receive}{\receive}{\(\longrightarrow\)}{}

\SetKwFunction{invalid}{invalid}
\SetKwFunction{nextNode}{nextNode}
\SetKwFunction{send}{send}

\SetKwData{coreR}{coreR}
\SetKwData{coreL}{\makebox[0pt][l]{coreL}\phantom{coreR}}
\SetKwData{threadR}{threadR}
\SetKwData{threadL}{\makebox[0pt][l]{threadL}\phantom{threadR}}

\SetKwProg{initial}{source node, initial action}{:}{}

\caption{BeRGeR Actions}\label{algBERGER}

\textbf{source node, input:} \\
\ \ \ \ \(t\) \tcp*[f]{target} \\
\ \ \ \ \(m\) \tcp*[f]{message}%
\vspace{1mm}

\initial(){}{ \label{lineInitAction}
    % align actions properly
    % \makebox[0pt] prints content without creating space for it
    % \phantom creates space for content without printing it 
    \makebox[0pt][l]{\( \send(m, s, t, \R, \bot, \langle \rangle ) \to \)}
    \( \phantom{\send(m, s, t, \R, k, \langle s \rangle) \to}
    \nextNode(t, s, t, \R, \bot) \)
    \tcp*[f]{new \R core}
        
    \makebox[0pt][l]{\( \send(m, s, t, \L, \bot, \langle \rangle) \to \)}
    \( \phantom{\send(m, s, t, \R, k, \langle s \rangle) \to}
    \nextNode(t, s, t, \L, \bot) \)
    \tcp*[f]{new \L core}
        
    \phantom{\(\send(m, s, t, \R, k, \langle s \rangle) \to\)}
    \(\makebox[0pt][r]{k := }
    \nextNode(t, s, t, \R, \bot)\)

    % this is the longest line, so is used inside \phantom{} in the other lines
    \(\send(m, s, t, \R, k, \langle s \rangle) \to\) \(\nextNode(t, s, t, \R, k)\)
    \tcp*[f]{new \R thread}

    \phantom{\(\send(m, s, t, \R, k, \langle s \rangle) \to\)}
    \(\makebox[0pt][r]{k := }
    \nextNode(t, s, t, \L, \bot)\)

    \makebox[0pt][l]{\( \send(m, s, t, \L, k, \langle s \rangle) \to \)}
        \( \phantom{\send(m, s, t, \R, k, \langle s \rangle) \to}
        \nextNode(t, s, t, \L, k) \)
        \tcp*[f]{new \L thread}
}
\vspace{1mm}

\textbf{every node:} \\
\Receive{\((m, s, t, c, k, \ell)\) \from \(p\)}{ \label{lineReceiveAction}
    \If{\(\invalid(p, m, s, t, k, \ell)\)}{
        \Return \tcp*[f]{drop invalid packets}
    }
    \If(\tcp*[f]{if this is a core}){\(k = \bot\)}{
        \(\ell\).append(\(p\))  \tcp*[f]{append sender to list of visited nodes}
    }

    \If(\tcp*[f]{packet is not at target}){\( n \neq t \)}{ \label{linePacketNotTarget}
                \(\send(m, s, t, c, k, \ell) \to \nextNode(p, s, t, c, k)\) \tcp*[f]{forward packet} \\
                \If(\tcp*[f]{if received packet is core}){\(k = \bot\)}{
                    \(k := \nextNode(p, s, t, c, \bot)\) \tcp*[f]{find next green node} \\
                    \If(\tcp*[f]{if \(k\) is unvisited}){\(k \not\in \ell\)}{ \label{lineCheckKUnvisited}
                        \(\send(m, s, t, c, k, \langle n \rangle) \to \nextNode(p, s, t, c, k)\) \tcp*[f]{new thread} \\
                    }
                }
    }

    \Else(\tcp*[f]{packet is at target}){\label{linePacketTarget}
        \tcp{green nodes neighboring \(t\)}
        \( \coreR := \nextNode(s, s, t, \L, \bot) \) \\
        \( \coreL := \nextNode(s, s, t, \R, \bot) \) \\
        \vspace{1mm}

        \tcp{node(s) neighboring \(t\) next to green nodes; may be identical}
        \( \threadR := \nextNode(s, s, t, \L, \coreR) \) \\
        \( \threadL := \nextNode(s, s, t, \R, \coreL) \) \\
        \vspace{1mm}

        \If{\( \ \ \ \ \ \ k = \bot \and c = \R \and p = \coreR \)}{ \label{lineIfRCore}
            $T$.add\((m, s, \R, \bot, \ell)\)  \tcp*[f]{record \R core}
        }
        \ElseIf{\( k \neq \bot \and c = \R \and  p \in \{\coreR, \threadR\} \)}{
            $T$.add\((m, s, \R, k, \langle \rangle)\)  \tcp*[f]{record \R thread}
        }
        \ElseIf{\( k = \bot \and c = \L \and p = \coreL \)}{
            $T$.add\((m, s, \L, \bot, \ell)\)  \tcp*[f]{record \L core}
        }
        \ElseIf{\( k \neq \bot \and c = \L \and p \in \{\coreL, \threadL\} \)}{
            $T$.add\((m, s, \L, k, \langle \rangle)\)  \tcp*[f]{record \L thread}
        }
        \vspace{1mm}

        \If{\( \exists \ell_1, \ell_2 : \{(m, s, \L, \bot, \ell_1), (m, s, \R, \bot, \ell_2) \} \subset T \and \ell_1 \cap \ell_2 = \{s\} \)}{ \label{lineMatchingCores}
            \(\deliver m\)  \tcp*[f]{matching cores}
        }
        \ElseIf{   \( \exists (m, s, c, \bot, \ell) \in T : \forall i \in \ell, (m, s, c, i, \langle \rangle) \in T \)}{ \label{lineBraidMatchesCore}
            \( \deliver m\) \tcp*[f]{matching core and braid of threads}
        }
    }
}
\end{algorithm}

%{\color{green} \(s\) and \(t\) are considered to be on the same side of \(\overline{st}\)}

\subsection{Algorithm details}
BeRGeR pseudocode is shown in Algorithms~\ref{algBERGERfunctions} and~\ref{algBERGER}.
Algorithm~\ref{algBERGERfunctions} shows constants, variables and functions used in BeRGeR. Specifically, each node maintains its own geometric coordinates \(n\) and the set of coordinates of its neighbors \(N\). Constants \L and \R denote packet traversal direction.
The target node \(t\) maintains a set of packets \(T\) that it received. The target checks \(T\) to see if the received packets satisfy message delivery conditions. Function \(\nextNode(p,s,t,c,k)\)
uses the neighbor \(p\) of \(n\) to select the node according to traversal direction \(c\), either \L or \R. This selection excludes node \(k\) if it is skipped by a thread, and it considers only the neighbors on the same side of \(\overline{st}\). This includes \(s\) and \(t\) themselves. This check is done in function \(\onThisSide(s,t,n,i)\). This function ensures that algorithm packets do not use the edges that cross the green face.

\ \\
Algorithm~\ref{algBERGER} shows BeRGeR actions. There are two: the source node initial action that originates the packets (see~\autoref{lineInitAction}), and the receipt action taken by every node when it receives a packet (see~\autoref{lineReceiveAction}). As input, the source node \(s\) takes target coordinates \(t\)and the message \(m\) to be delivered to it.
In the initial action, the source sends four packets: two cores and two threads. The cores go in left and right directions along the green face. The threads skip the first green nodes in the two directions.

The packet format is as follows: the message \(m\); the source, $s$, and target, $t$; traversal direction \L or \R; the node \(k\) to be skipped or \(\bot\) if it is a core packet; and the list of visited nodes. A core packet starts with an empty list. This list is only updated for cores.
A thread carries its originator in the visited list. 

Let us describe the packet receive action. First, the packet is checked for validity (see~\autoref{lineInvalid}). The packet is dropped if it has traveled in a cycle, a core is passing through \(s\), or the sender, \(p\), is sending a thread that skips \(p\). Otherwise, if the received packet is a core, \(p\) is appended to the visited node list \(\ell\) and further processing depends on whether the packet has arrived at the target.

If the packet recipient is not the target (see~\autoref{linePacketNotTarget}), then the node forwards the packet and, if it is a core packet, then the node also sends a thread skipping the next green node, provided that green node is not in the visited list \(\ell\).

If the packet recipient node is the target (see~\autoref{linePacketTarget}), the recipient checks that the packet comes from an expected node and then records the receipt of the core or thread packet in \(T\). Then, the target determines if message delivery conditions are met. Specifically, if $T$ has a record of two matching cores or a matching core and a braid. \emph{Matching cores} carry the same message in the opposite traversal directions. \emph{Matching core and braid} are a core and a set of threads such that they are in the same traversal direction, carry the same message and, for every node that the core visited, there is a thread that skips it. 
If the target receives matching cores or a matching core and a braid, the target delivers the message that they carry.

\section{BeRGeR Correctness Proof}

\begin{lemma}\label{lemSingleFace}
    A core packet traverses a single face of $G - \overline{st}$ and a thread skipping node $k$ traverses a single face of $G - \overline{st} -\{k\}$.
\end{lemma}
\begin{proof}
    By the design of the algorithm, a packet traverses a single face. In forwarding a core packet, each node ignores the edges that intersect $\overline{st}$. If the source originates a core packet, this packet traverses the green face of $G - \overline{st}$. Similarly, thread skipping $k$ originated by the source traverses the union of the green and the $k$-blue face of $G - \overline{st} -\{k\}$. A faulty node may send a packet in some other face, in which case it traverses that face. 
    \qed
\end{proof}

\begin{figure}
\centering
\includegraphics[angle=-90,width=.30\linewidth]{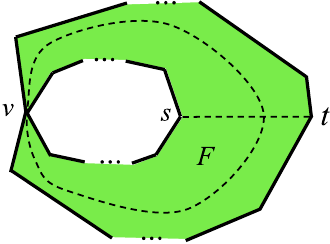}
\caption{Illustration for the proof of \autoref{coreDisjoint}. If left and right core paths share node $v$, then $v$ can be connected by a continuous curve that separates $s$ from $t$. Hence, every path from $s$ to $t$ contains $v$.}\label{figDisjoint}
\end{figure}

Let \emph{left core path} be the left-hand-rule traversal path on the green face from the source to the target. Similarly, \emph{right core path} be the right-hand-rule traversal path on the green face from $s$ to $t$.

\begin{lemma}\label{coreDisjoint}
Left and right core paths are internally node-disjoint.
\end{lemma}

\begin{proof}
We prove the lemma by contradiction. Suppose the opposite: the left and right core paths share an internal node $v$. See \autoref{figDisjoint} for illustration.
In this case, we can draw a closed curve that starts and ends in $v$ and whose interior is inside the green face. This curve separates the plane into two areas: one of them contains $s$ and the other $t$. This means that every path from $s$ to $t$ contains $v$.
However, \nameref{G-st_3connected} states
that $G - \overline{st}$ is three-connected. This means that there must be at least three 
internally node-disjoint paths between $s$ and $t$. 
Hence, our initial supposition is not correct. Therefore, left and right core paths must be internally node-disjoint.
\qed
\end{proof}

Let \emph{left core} and \emph{right core} are packets following left and right core paths respectively.

\begin{lemma}[Core validity]\label{coreValidity}
    If the target receives a left core and a right core carrying messages $m_{lc}$ and $m_{rc}$ respectively and $m_{lc} = m_{rc}$, then $m_{lc} = m_s$.
\end{lemma}
\begin{proof}
The target receives a core packet from a node adjacent to the green face. According to~\autoref{lemSingleFace}, such packet traverses the green face only. According to~\autoref{coreDisjoint}, left and right core paths are disjoint.  Since there is at most one fault in the network, at least one of these paths is fault-free. Therefore, either left or right core packets are forwarded by correct nodes only. In this case, it carries the message sent by the source.  Hence, if the target receives two identical messages from both left and right core paths, this message is sent by the source.\qed
\end{proof}

% TODO: add the second case to description -MN
% TODO: differentiate between ORIGINATE and SPLIT !!
\begin{lemma}[Thread validity]\label{lemThread}
If a thread skips a green node $k$ and reaches a correct node, it is not originated by $k$.
\end{lemma}
\begin{proof}
% TODO: could we use f instead of k here? -Zaz
Only the source and the faulty node may originate threads. We assume that the source is correct. Let us consider some {other} faulty node $k$. 
Note that $k$ may potentially originate a thread that skips $k$, or may originate a core that splits into a thread that skips $k$. If $k$ originates a thread that skips itself, then, by the design of the algorithm (see \autoref{lineSkipsItself}), the recipient does not forward it. 

Consider now the case where $k$ originates a core. In this case, this core contains $k$ in the list of visited nodes $\ell$. Therefore, no correct node that receives this core sends a thread that skips $k$ (see~\autoref{lineCheckKUnvisited}).  In either case, the thread that skips $k$ is not originated by $k$. 
\qed
\end{proof}

\begin{comment}
% no longer needed: the info is now in Lemma 1
Recall that given a green node $k$, a $k$-blue  is a union of all non-green faces adjacent to $k$.  A thread that skips $k$, traverses this face and the green face. Alternatively, a thread that skips $k$ traverses the face that contains $\overline{st}$ in  $G - \overline{st} - \{k\}$ subgraph.
\end{comment}
\begin{figure}
\centering

\includegraphics[angle=-90,width=.40\linewidth]{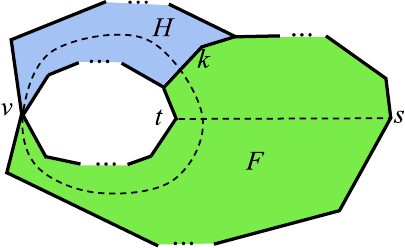}
% TODO: "blue face $G$" is confusing. Could we use "B" instead? -Zaz
% TODO: If k is in the correct place, G is not the k-blue face. The k-blue face would intersect the green face at the nodes adjacent to k
% fixed, MN
\caption{Illustration for the proof of \autoref{threadDisjoint}. If a blue face, $H$, bypassing green node $k$ contains a green node, $v$, that lies on the right core path, then $v$ can be connected by a continuous curve inside $F\cup H$ that separates $s$ from $t$ and, hence, every path from $s$ to $t$ either contains $k$ or $v$.}\label{figKDisjoint}
\end{figure}

\begin{lemma}\label{threadDisjoint}
The path traversed by a left thread does not contain any node of the right core path. Similarly, the path traversed by the right thread does not contain nodes of the left core path.
\end{lemma}
\begin{proof}
    We prove the lemma for the left thread. The argument for the right thread is similar. 
   
    Let us consider a left thread that skips some node $k$.  According to~\autoref{lemSingleFace}, the thread either traverses the left core path or the $k$-blue face.
    Due to~\autoref{coreDisjoint}, the right and left core paths are internally node-disjoint. Hence, the left core path does not contain right core path nodes. 
    
    Let us now consider the nodes adjacent to the $k$-blue face.
    We prove this part by contradiction. Suppose the opposite: there is a node $v$ that is adjacent to the $k$-blue face but lies on the right core path. We show that in this case all paths from $s$ to $t$ contain either $v$ or $k$. Specifically, we show that the path that does not contain $k$, goes through $v$. 
    
    Indeed, removing node $k$ and adjacent edges merges the green face and the $k$-blue face. See~\autoref{figKDisjoint} for illustration.  In this figure, $F$ is the green face and $H$ is the $k$-blue face. In this case, inside this joint face, we can draw a closed curve that starts and ends in $v$. Therefore, any path from $s$ to $t$ that does not contain $k$, has to go through $v$. That is, all paths from $s$ to $t$ contain either $v$ or $k$. 
 
    However, \nameref{G-st_3connected} requires that there are three internally node-disjoint paths from $s$ to $t$. Hence, our supposition is incorrect and the path traversed by the left thread, either along the left core path or along the $k$-blue face, does not contain a node on the right core path.
    \qed
\end{proof}

Recall, that a \emph{braid} is a set of threads with the same direction, left or right, and carrying the same message.
A braid \emph{matches} a core packet $c$ if (i) $c$ and the braid carry the same message and (ii) for every node $i$ that $c$ carries in its visited list $\ell$ there is a thread in the braid that skips $i$.  

\begin{lemma}[Braid validity] \label{lemBraid}
If the fault is adjacent to the green face and the target receives a core and a matching braid carrying $m_b$, then $m_b = m_s$.
\end{lemma}
\begin{proof}
We prove the lemma for the left direction. The argument for right direction is similar.
The fault may be on the left side of the green face or on the opposite side. We consider these cases separately.

Let the target receive a left core and matching left braid while the fault is on the right side. In this case, according to~\autoref{coreDisjoint}, all the nodes of the left core path are correct. By the design of the algorithm (see \autoref{lineIfRCore}), the target accepts a core packet only if it comes from a neighbor adjacent to the green face. Due to~\autoref{lemSingleFace}, this core traverses the green face only. Since the fault is on the right core path, this left core was forwarded by correct nodes only. Therefore it carries $m_s$. 

Let us now consider the case where the target receives a left core and a matching left braid while the fault is also on the left core path. Let $f$ be the faulty node; that is, the faulty node is on the received core path. 
A correct recipient that forwards the packet records the packet sender in the visited list $\ell$ of the packet. Therefore, $f$ is present in $\ell$ of the core packet that the target receives. Since the target receives a braid that matches this core packet, it also receives a thread that skips $f$. According to~\autoref{lemThread}, $f$ may not originate such a thread. Hence, the thread that skips $f$ is originated by the source so it carries $m_s$. Since this thread is in the braid that carries the same message and matches the core message, they all carry $m_s$.  \qed
\begin{comment}
\\
% TODO: Decide if you want to simplify the above. If not, feel free to delete this. -Zaz
{\color{green}
\textbf{Alternatively}

If we receive a matching braid and core, there are two possibilities. Either the core did not pass through the faulty node, or it did. We consider these cases separately.

If the core did not pass through the faulty node, the core must be correct.

If the core did pass through the faulty node, \(f\):
A correct recipient that forwards the packet, records the packet sender in the visited list $\ell$ of the packet. Therefore, $f$ is present in $\ell$ of the core packet that the target receives. Since the target receives a braid that matches this core packet, it also receives a thread that skips $f$. According to~\autoref{lemThread}, $f$ may not originate such a thread. Hence, the thread that skips $f$ is originated by the source. Therefore, it carries $m_s$. Since this thread is in the braid that carries the same message and matches the core message, they all carry $m_s$.  
\qed
}
\end{comment}
\begin{comment}
{\color{green} or originate a core that later generates a thread that skips \Byz}.
% TODO: the source did not necessarily originate $d$ -Zaz
That is{\color{gray}, the source originated $d$} and $d$ did not go through $f$.
 Therefore, thread $d$ was forwarded by correct nodes only. Thus, $d$ carries message $m_s$ originated by the source. 
Since all threads in the matching braid, as well as the matching core packet, carry the same message, they all carry $m_s$.\qed
\end{comment}
\end{proof}

\begin{lemma}[Liveness]\label{liveness}
 The target eventually receives either (i) matching left and right core packets or (ii) a matching core and a braid. 
\end{lemma}
\begin{proof}
  If the faulty node is not adjacent to the green face, then both matching left and right core packet reach the target. 

  Let us examine the case of a faulty node is adjacent to the green face. This means that it lies on a core path. Assume, without loss of generality, that the node is on the right core path. Then, according to~\autoref{coreDisjoint}, the left core path is fault-free. Moreover, due to~\autoref{threadDisjoint} the paths of all the left threads are also fault-free. That is, in case the faulty node is on the right core path, the left core and a matching left braid reach the target. 
\end{proof}

\begin{lemma}[Termination]\label{termination}
     Every packet is forwarded a finite number of times.
\end{lemma}
\begin{proof}
% TODO: the source does not drop threads, but this does not affect termination. The algorithm still terminates even without anything special happening at s or t, so I suggest we remove the below sentence -Zaz
% ok, MN
% The target, the faulty node{\color{red}, and the source} drop packets if they receive them. 
According to~\autoref{lemSingleFace}, every packet traverses a single face in either $G-\overline{st}$ or $G - \overline{st} - \{k\}$. 
The originator of the packet is recorded in the visited list \(\ell\) of the packet. When received, the originating node drops such a packet. That is, if the packet is not dropped by the source, target or the faulty node, it is dropped by the originating node once the packet traverses the entire face.

The faulty may record a spurious packet originator in \(\ell\). However, according to~\autoref{lemSingleFace}, this packet traverses some face, arrives back to the faulty node, where it is assumed to be dropped.

In a finite graph, this means that every packet is forwarded a finite number of times.
\qed
\end{proof}

\begin{theorem}\label{thrmBerger}
BeRGeR: Byzantine Robust Geometric Routing algorithm solves the Reliable Message Delivery Problem with a single Byzantine node in a planar graph subject to \nameref{G-st_3connected}.
\end{theorem}
\begin{proof}
In BeRGeR, the target delivers message $m_t$ in two cases: either it receives two matching core packets or it receives a matching core packet and a braid. According to~\autoref{coreValidity}, if the target receives matching core packets, the packets carry the message sent by the source. According to~\autoref{lemBraid}, if the target receives a matching core and a braid, they carry the source message also. That is, the target delivers only the message sent by the source. Hence, BeRGeR satisfies the Validity Property of the Reliable Message Delivery Problem. 

Moreover,~\autoref{liveness} guarantees that the target eventually receives matching cores or a matching core and a braid. That is, BeRGeR guarantees that a message is delivered by the target, which means that the algorithm also satisfies the Liveness Property. 

\autoref{termination} shows that BeRGeR satisfies the Termination Property. \qed
\end{proof}

\section{Constant Packet Size Extension~and Complexity~Estimate}

\subsection{Constant packet size extension}
In BeRGeR, a core packet carries the path that it travels, making the packet size potentially linear with respect to the network size. However, the modification to constant size cores is relatively straightforward.
This modification is achieved at the expense of stateless packet forwarding.
For that, the sender transmits the message in numbered fixed-size packets to the neighbor. The neighbor receives the packets and reassembles the message. If any of the packets are missing, the whole message is discarded. Since we assume no packet loss, a correct node transmits all packets. The faulty node either transmits the packets or fails to do so. The letter is equivalent to no packet transmission of the original algorithm. The correctness of the original BeRGeR is preserved.

\subsection{Algorithm message complexity}
Let us analyze the message complexity of BeRGeR during its fault-free operation. Let \(N\) and \(E\) be the respective number of nodes and edges in graph \(G\). In the core path, each node may appear at most once. Hence, the length of the core path is in \(O(N)\).  Each core message transmission may be separated into \(O(N)\) constant-size packets. Since the path of the core message is in \(O(N)\), the total number of packets for a single
core is in \(O(N^2)\).

For each node on the core path, the algorithm generates a thread. Each thread travels at most \(E\) edges and the threads are constant size. Therefore, the number of thread packet transmissions generated by a single core is in \(O(EN)\). There are left core and right core. Hence, the total number of sent packets is in \(O(2N^2 + 2EN)\). In a planar graph, \(E \in O(N)\) by Euler's formula.  Thus, the overall BeRGeR message complexity is in \(O(N^2)\).

\section{Future Work}

As described in this paper, BeRGeR is a unicast algorithm. To achieve Byzantine fault tolerance, it employs the same message concurrency techniques that are used to solve geocasting~\cite{adamek2017stateless} and multicasting~\cite{adamek2018concurrent}. We expect BeRGeR can be adapted in a straightforward manner to produce Byzantine-robust solutions to these two problems. 

BeRGeR requires \nameref{G-st_3connected} to operate correctly. We are unsure if byzantine-robust geometric routing can be achieved with polynomial message complexity without it. Trying to relax it would be an interesting research pursuit.

Byzantine-robust routing needs the communication graph to be three-connected. The maximum planar graph connectivity is \(5\). Thus, potentially, such a graph may enable a geometric routing algorithm that can tolerate up to \(2\) Byzantine faults. Finding such an algorithm is another research challenge.

\section*{Acknowledgments}

We would like to thank Sam Kosco for many helpful discussions leading to the algorithm in its present form.

%for being the bane of our existence and discovering counterexamples that showed our original algorithm not to work without non-constant message size and a non-constant number of braids.

\bibliographystyle{unsrt} 
\bibliography{bib}

\end{document}